\documentclass[conference]{IEEEtran}
\IEEEoverridecommandlockouts
\usepackage{cite}
\usepackage{amsmath,amssymb,amsfonts}
\usepackage{algorithm,algorithmic}
\usepackage{graphicx}
\usepackage{textcomp}
\def\BibTeX{{\rm B\kern-.05em{\sc i\kern-.025em b}\kern-.08em
		T\kern-.1667em\lower.7ex\hbox{E}\kern-.125emX}}
	
\usepackage{lipsum}
\usepackage{amsmath,amssymb,amsthm,mathrsfs,amsfonts,dsfont}
\DeclareMathOperator*{\argmax}{arg\,max}

\usepackage{tabularx}
\usepackage{subfigure}
\usepackage{xcolor}
\usepackage{epsfig}
\usepackage{epstopdf}
\usepackage{adjustbox}

\newtheorem{proposition}{Proposition}

\newtheorem{remark}{Remark}

\newtheorem{lemma}{Lemma}
\begin{document}
	
\title{Optimal Water-Filling Algorithm in Downlink Multi-Cluster NOMA Systems\\
}

\author{
	\IEEEauthorblockN{Sepehr Rezvani\IEEEauthorrefmark{1}, and Eduard A. Jorswieck\IEEEauthorrefmark{1}}
	\IEEEauthorblockA{\IEEEauthorrefmark{1}Institute for Communications Technology, Technische Universität Braunschweig, Braunschweig, Germany}
}
	
\maketitle

\begin{abstract}
	 The key idea of power-domain non-orthogonal multiple access (NOMA) is to exploit the superposition coding (SC) combined with successive interference cancellation (SIC) technique (called SC-SIC) while reducing the receivers' complexity as well as error propagation. Actually, NOMA suggests a low-complexity scheme, where users are grouped into multiple clusters operating in isolated resource blocks, and SC-SIC is performed among users within each cluster. In this paper, we propose a globally optimal joint intra- and inter-cluster power allocation for any arbitrary user grouping to maximize users' sum-rate. In this algorithm, we exploit the closed-form of optimal intra-cluster power allocation obtained in our previous work. Then, by transforming network-NOMA to an equivalent virtual network-OMA, we show that the optimal power allocation can be obtained based on the very fast water-filling algorithm. Interestingly, we observe that each NOMA cluster acts as a virtual OMA user whose effective channel gain is obtained in closed form. Also, each virtual OMA user requires a minimum power to satisfy the minimum rate demand of its real multiplexed user. In simulation results, we evaluate the performance gap between fully SC-SIC, NOMA, and OMA, in terms of users sum-rate, and outage probability.
\end{abstract}	


\section{Introduction}
\allowdisplaybreaks
\IEEEPARstart{T}{the} channel capacity of single-input single-output (SISO) Gaussian broadcast channels (BCs) can be achieved by superposition coding (SC) at the transmitter, and coherent multiuser detection algorithms, like successive interference cancellation (SIC) at the receivers \cite{PHDdiss,NIFbook}. Although fully SC-SIC, i.e., multiplexing all the receivers' signal, is capacity-achieving in SISO Gaussian BCs that are known to be degraded \cite{PHDdiss,NIFbook}, it results in high complexity at the receivers for performing SIC. Also, the error propagation limits the practical implementation of fully SC-SIC for larger number of receivers \cite{6692652,7263349}. 
To tackle the mentioned challenges, the work in \cite{6692652} proposed a low-complexity technique, called non-orthogonal multiple access (NOMA), which is a combination of SC-SIC and orthogonal multiple access (OMA) such as frequency-division multiple access (FDMA) or time-division multiple access (TDMA). In NOMA, the users are grouped into multiple isolated clusters, and SC-SIC is performed among users within each cluster. The clusters operate in isolated resource blocks, based on FDMA/TDMA, thus the inter-cluster interference is eliminated\footnote{The spatial domain multiple access (SDMA) can also be introduced on NOMA, where clusters are isolated by zero-forcing beamforming \cite{7263349}.}. Since each user belongs to only one NOMA cluster, each user occupies only one resource block. Accordingly, each resource block can be viewed as a SISO Gaussian BC whose capacity is achieved by SC-SIC. In doing so, NOMA is considered as one of the promising candidate solutions for the 5th and 6th generations wireless networks \cite{9154358}.

Although SISO Gaussian BCs are degraded, meaning that successful SIC among multiplexed users does not depend on power allocation \cite{8823873}, it is well-known that power allocation in NOMA is essential to achieve preferable performance \cite{7263349,9154358,8823873}. The optimal power allocation in NOMA includes two components: 1) optimal power allocation among multiplexed users within each cluster; 2) optimal power allocation among these isolated clusters (optimal power budget of clusters). In our previous work \cite{rezvani2021optimal}, we showed that the Karush–Kuhn–Tucker (KKT) optimality conditions analysis for fully SC-SIC, also called single-cluster NOMA, results in a unique closed-form of optimal powers for both the sum-rate maximization and power minimization problems. In this work, we consider the general NOMA system with multiple clusters each having any arbitrary number of users. To the best of our knowledge, the only work that addressed the optimal joint intra- and inter-cluster power allocation is \cite{7982784}. In \cite{7982784}, each group has only two users, and all the analysis is based on allocating more power to each weaker user to ensure successful SIC, which is not necessary (Myth 1 in \cite{8823873}). That is, the optimal power allocation in downlink (multi-cluster) NOMA for maximizing users sum-rate under the individual users minimum rate demand is still an open problem. 
Our main contributions are listed as follows:
\begin{itemize}
	\item We show that for any given power budget of clusters, only the cluster-head (strongest) users, which have the highest decoding order, deserve additional power. The rest of the users get power to only maintain their minimal rate demands.
	\item We show that $N$-cluster NOMA can be equivalently transformed to a $N$-virtual user OMA system. Each cluster is modeled as a virtual OMA user, whose effective channel gain is obtained in closed form. The minimum rate demand of users within each cluster is modeled as the minimum required power of the virtual user. Therefore, each virtual user requires a minimum power to satisfy the quality-of-service (QoS) demand of multiplexed users.
	\item After the transformation of NOMA to its equivalent virtual OMA system, we obtain the globally optimal power allocation by using a very fast water-filling algorithm. More importantly, this algorithm is valid for heterogeneous number of multiplexed users within each cluster, so it can be used for any arbitrary user grouping.
	\item We provide a conceptual performance comparison between three main schemes: fully SC-SIC, NOMA (with various maximum number of multiplexed users), and OMA. The performance comparison between NOMA and fully SC-SIC brings new insights on the suboptimality-level of NOMA due to user grouping based on FDMA/TDMA. In this work, we answer the question \textit{"How much performance gain can be achieved if we increase the order of NOMA clusters, and subsequently decrease the number of user groups?"} for a wide range of number of users and their minimum rate demands. The latter knowledge is highly necessary since multiplexing a large number of users would cause high complexity cost at the users' hardware, due to the SIC and existing error propagation \cite{7263349,9154358}.
\end{itemize} 

The rest of this paper is organized as follows: Section \ref{Section sysmodel} describes the general system model of downlink single-cell NOMA. Our proposed optimal water-filling algorithm for solving the sum-rate maximization problem in NOMA is presented in Section \ref{Section solution}. Numerical results are presented in Section \ref{Section simulation}. Our conclusions and future research directions are presented in Section \ref{Section conclusion}.

\section{Downlink Single-Cell NOMA}\label{Section sysmodel}
Consider the downlink channel of a multiuser system, where a BS serves $K$ users with limited processing capabilities within a unit time slot. The duration of each time slot is small such that the channel is time-invariant within the time slot. The set of users is denoted by $\mathcal{K}=\{1,\dots,K\}$. The maximum number of multiplexed users is $U^\text{max}$. Hence, the fully SC-SIC technique is infeasible when $U^\text{max} < K$. In NOMA, users are grouped into $N$ NOMA clusters for SC-SIC\footnote{In NOMA, the condition $U^\text{max} < K$ implies that $N \geq 2$.}. 
Each NOMA cluster operates in an isolated subchannel based on FDMA. We assume that the total bandwidth $W$ (Hz) is equally divided into $N$ isolated subchannels, where the bandwidth of each subchannel is $W_s=W/N$ \cite{7982784}. The set of subchannels is denoted by $\mathcal{N}=\{1,\dots,N\}$. The set of multiplexed users on subchannel $n$ is denoted by $\mathcal{K}_n$. In NOMA, each user belongs to only one cluster \cite{6692652,7982784}. As a result, we have $\mathcal{K}_n \cap \mathcal{K}_m = \emptyset, \forall n,m \in \mathcal{N},~ n \neq m$.
Moreover, we have $|\mathcal{K}_n| \leq U^\text{max},~\forall n \in \mathcal{N}$, where $|.|$ indicates the cardinality of a finite set.

Each subchannel $n$ can be modeled as a SISO Gaussian BC. The received signal at user $k \in \mathcal{K}_n$ on subchannel $n$ is
\begin{equation}\label{received signal}
	y^n_{k}= \sum\limits_{i \in \mathcal{K}_n}
	\sqrt{p^n_{i}} g^n_{k} s^n_{i}= \underbrace{\sqrt{p^n_{k}} g^n_{k} s^n_{k}}_\text{intended signal} + \underbrace{\sum\limits_{i \in \mathcal{K}_n\setminus\{k\}} \sqrt{p^n_{i}} g^n_{k} s^n_{i}}_\text{co-channel interference} + z^n_{k},
\end{equation}
where $s^n_{k} \sim \mathcal{CN} (0,1)$ and $p^n_{k} \geq 0$ are the modulated symbol from Gaussian codebooks, and transmit power of user $k$ on subchannel $n$, respectively. Obviously, $p^n_{k}=0, \forall k \notin \mathcal{K}_n$. Moreover, $g^n_{k}$ is the (generally complex) channel gain
from the BS to user $k$ on subchannel $n$, and $z^n_{k}\sim \mathcal{CN} (0,\sigma^n_{k})$ is the additive white Gaussian noise (AWGN).
We assume that perfect channel state information (CSI) of all the users is available at the schedular \cite{rezvani2021optimal,7982784}.

Based on the NOMA protocol, SC-SIC is applied to each multiuser subchannel. Independent from the power allocation among all the users, the channel-to-noise (CNR)-based decoding order is optimal in each subchannel \cite{NIFbook,6692652,7263349,9154358,8823873,rezvani2021optimal,7982784}. Let $h^n_{k}=|g^n_{k}|^2/\sigma^n_{k},~\forall n \in \mathcal{N}, k \in \mathcal{K}_n$. 
Then, the decoding order based on $h^n_{i} > h^n_{j} \Rightarrow i \to j,\forall i,j \in \mathcal{K}_n$ is optimal, where $i \to j$ represents that user $i$ fully decodes (and then cancels) the signal of user $j$ before decoding its desired signal on subchannel $n$. We call the stronger user $i$ as the user with higher decoding order in the multiplexed user pair $i,j \in \mathcal{K}_n$. The SIC protocol in each isolated subchannel is the same as the SIC protocol of fully SC-SIC. 
In the SIC of NOMA, the user with the highest decoding order, also called NOMA cluster-head user, decodes and cancels the signal of all the other multiplexed users in each subchannel. Hence, the NOMA cluster-head users do not experience any interference. For each subchannel $n$, the index of NOMA cluster-head user is denoted by $\Phi_n=\argmax\limits_{k \in \mathcal{K}_n} h^n_{k}$. For the case that a subchannel is assigned to only one user, i.e., $|\mathcal{K}_n|=1$, this single user which does not experience any interference is indeed the NOMA cluster-head user on that subchannel. That is, NOMA is a combination of SC-SIC and OMA. After successful SIC, the signal-to-interference-plus-noise ratio (SINR) of user $k$ for decoding its own signal $s^n_{k}$ is $\gamma^n_{k} = \frac{p^n_{k} h^n_{k}}{\sum\limits_{j \in \mathcal{K}_n, \atop h^n_{j} > h^n_{k}} p^n_{j} h^n_{k} + 1}$. Subsequently, the achievable rate (in bps) of user $k \in \mathcal{K}_n$ based on the Shannon's capacity formula is given by \cite{NIFbook}
\begin{equation}\label{achiev se}
	R^n_{k} (\boldsymbol{p}^n) = W_s \log_2\left( 1+ \gamma^n_{k} (\boldsymbol{p}^n) \right),
\end{equation}
where $\boldsymbol{p}^n=[p^n_{k}],~\forall k \in \mathcal{K}$ is the vector of allocated powers to all the users on subchannel $n$.
Assume that the set of NOMA clusters, i.e., $\mathcal{K}_n,~\forall n \in \mathcal{N}$ is pre-defined \cite{7982784}. The general power allocation problem for maximizing users sum-rate under the individual minimum rate demands of users in NOMA is formulated by
\begin{subequations}\label{pureNOMA problem}
	\begin{align}\label{obf pureNOMA problem}
		\max_{ \boldsymbol{p} \geq 0}\hspace{.0 cm}	
		~~ & \sum\limits_{n \in \mathcal{N}} \sum\limits_{k \in \mathcal{K}_n} R^n_{k} (\boldsymbol{p}^n)
		\\
		\text{s.t.}~~
		\label{Constraint minrate}
		& R^n_{k} (\boldsymbol{p}^n) \geq R^{\text{min}}_{k},~\forall n \in \mathcal{N}, k \in \mathcal{K}_n,
		\\
		\label{Constraint cell power}
		& \sum\limits_{n \in \mathcal{N}} \sum\limits_{k \in \mathcal{K}_n} p^n_{k} \leq P^{\text{max}},
		\\
		\label{Constraint mask power}
		& \sum\limits_{k \in \mathcal{K}_n} p^n_{k} \leq P^{\text{mask}}_n, \forall n \in \mathcal{N},
	\end{align}
\end{subequations}
where 
\eqref{Constraint minrate} is the per-user minimum rate constraint, in which $R^{\text{min}}_{k}$ is the individual minimum rate demand of user $k$. \eqref{Constraint cell power} is the cellular power constraint, where $P^{\text{max}}$ denotes the maximum available power of the BS. \eqref{Constraint mask power} is the maximum per-carrier power constraint, where $P^{\text{mask}}_n$ denotes the maximum allowable power on subchannel\footnote{We do not impose any specific condition on $P^{\text{mask}}_n$. We only take into account $P^{\text{mask}}_n$ in our analysis to keep the generality, such that $P^{\text{mask}}_n=P^{\text{max}},~\forall n \in \mathcal{N}$ as special case.} $n$. For convenience, we denote the general power allocation matrix as $\boldsymbol{p}=[\boldsymbol{p}^n],\forall n \in \mathcal{N}$.

\section{Water-Filling for Sum-Rate Maximization Problem}\label{Section solution}
Here, we propose a water-filling algorithm to find the optimal solution of \eqref{pureNOMA problem}.
\begin{lemma}\label{lemma convex NOMA}
	Problem \eqref{pureNOMA problem} is convex in $\boldsymbol{p}$.
\end{lemma}
\begin{proof}
	The sum-rate of users in each NOMA cluster, i.e., $\sum\limits_{k \in \mathcal{K}_n} R^n_{k} (\boldsymbol{p}^n)$, is indeed the sum-capacity of each subchannel $n$ whose entire region is convex \cite{NIFbook}. Another proof is calculating the Hessian of the sum-rate function which is proved to be negative definite, so the sum-rate is strictly concave. Due to the isolation among subchannels, the overall sum-rate in \eqref{obf pureNOMA problem} is strictly concave in $\boldsymbol{p}$. Besides, the power constraints in \eqref{Constraint cell power} and \eqref{Constraint mask power} are affine, so are convex. The minimum rate constraint in \eqref{Constraint minrate} can be equivalently transformed to the following affine form as
	$2^{(R^{\text{min}}_{k}/W_s)} \bigg(\sum\limits_{j \in \mathcal{K}_n, \atop h^n_{j} > h^n_{k}} p^n_{j} h^n_{k} + 1\bigg) \leq 
	\sum\limits_{j \in \mathcal{K}_n, \atop h^n_{j} > h^n_{k}} p^n_{j} h^n_{k} + 1 + p^n_k h^n_k,~\forall n \in \mathcal{N}, k \in \mathcal{K}_n.$
	Accordingly, the feasible set of \eqref{pureNOMA problem} is convex. Summing up, the problem \eqref{pureNOMA problem} is convex in $\boldsymbol{p}$.
\end{proof}
According to Lemma \ref{lemma convex NOMA}, the general sum-rate maximization problem in NOMA is convex. To find the globally optimal solution, we first define $q_n=\sum\limits_{k \in \mathcal{K}_n} p^n_{k}$ as the power consumption of NOMA cluster (or subchannel) $n$. The main problem \eqref{pureNOMA problem} can be equivalently transformed to the following joint intra- and inter-cluster power allocation problem as
\begin{subequations}\label{SCuser problem}
	\begin{align}\label{obf SCuser problem}
		\max_{\boldsymbol{p} \geq 0, \boldsymbol{q} \geq 0}\hspace{.0 cm}	
		~~ & \sum\limits_{n \in \mathcal{N}} \sum\limits_{k \in \mathcal{K}_n} R^n_{k} (\boldsymbol{p}^n)
		\\
		\text{s.t.}~~
		\label{Constraint minrate SCuser}
		& R^n_{k} (\boldsymbol{p}^n) \geq R^{\text{min}}_{k},~\forall n \in \mathcal{N}, k \in \mathcal{K}_n,
		\\
		\label{Constraint cell q}
		& \sum\limits_{n \in \mathcal{N}} q_n \leq P^{\text{max}},
		\\
		\label{Constraint sumpow}
		& \sum\limits_{k \in \mathcal{K}_n} p^n_{k} = q_n,~\forall n \in \mathcal{N},
		\\
		\label{Constraint mask q}
		& 0 \leq q_n \leq P^{\text{mask}}_n, \forall n \in \mathcal{N}.
	\end{align}
\end{subequations}
In the following, we first convert the feasible set of \eqref{SCuser problem} to the intersection of closed-boxes and the affine cellular power constraint.
\begin{lemma}\label{lemma feasiblecluster}
	The feasible set of \eqref{SCuser problem} is the intersection of $q_n \in \left[Q^\text{min}_n,P^{\text{mask}}_n\right],~\forall n \in \mathcal{N}$, and cellular power constraint $\sum\limits_{n \in \mathcal{N}} q_n \leq P^{\text{max}}$. The constant $
	Q^\text{min}_n=
	\sum\limits_{k \in \mathcal{K}_n} \beta_{k} \left(\prod\limits_{j \in \mathcal{K}_n \atop h^n_{j} > h^n_{k}} \left(1+\beta_{j}\right) +\frac{1}{h_k}+\sum\limits_{j \in \mathcal{K}_n \atop h^n_{j} > h^n_{k}} \frac{ \beta_{j} \prod\limits_{l \in \mathcal{K}_n \atop h^n_{k} < h^n_{l} < h^n_{j}} \left(1+\beta_{l}\right)}{h_j}\right)
	$ is the lower-bound of power budget $q_n$,
	in which $\beta_k=2^{(R^{\text{min}}_k/W_s)} -1,~\forall k \in \mathcal{K}$.
\end{lemma}
\begin{proof}
	The feasibility of \eqref{SCuser problem} can be determined by solving the power minimization problem as
	\begin{equation}\label{min SCuser2 problem}
		\min_{\boldsymbol{p} \geq 0, \boldsymbol{q} \geq 0}~\sum\limits_{n \in \mathcal{N}} q_n~~~~~~\text{s.t.}~\eqref{Constraint minrate SCuser}\text{-}\eqref{Constraint mask q}.
	\end{equation}
	In the following, we find the feasibility region of $\boldsymbol{p}$. Assume that the feasible region of \eqref{min SCuser2 problem} is non-empty. The problem \eqref{min SCuser2 problem} is also convex with affine objective function. The strong duality holds since \eqref{Constraint cell q} and \eqref{Constraint mask q} hold with strict inequalities. In the KKT optimality conditions analysis in Appendix C of \cite{rezvani2021optimal}, we proved that in fully SC-SIC, the maximum power budget does not have any effect on the optimal powers obtained in the power minimization problem. And, at the optimal point, all the multiplexed users get power to only maintain their minimal rate demands. In this system, if the strong duality holds, the problem can be equivalently decoupled into $N$ single-carrier power minimization problems. In each subchannel $n$, regarless of the clusters power budget $\boldsymbol{q}$, the optimal power of each user can be obtained in closed form given by 
	$
		{p^n_{k}}^*=\beta_{k} \left(\prod\limits_{j \in \mathcal{K}_n \atop h^n_{j} > h^n_{k}} \left(1+\beta_{j}\right) +\frac{1}{h_k}+\sum\limits_{j \in \mathcal{K}_n \atop h^n_{j} > h^n_{k}} \frac{ \beta_{j} \prod\limits_{l \in \mathcal{K}_n \atop h^n_{k} < h^n_{l} < h^n_{j}} \left(1+\beta_{l}\right)}{h_j}\right), \forall k \in \mathcal{K}_n,
	$
	where $\beta_k=2^{(R^{\text{min}}_k/W_s)} -1,~\forall k \in \mathcal{K}$. Since the solution of \eqref{min SCuser2 problem} is unique, it can be easily shown that the optimal powers are indeed component-wise minimum. It other words $q_n (\boldsymbol{p^*}^n) \leq q_n (\boldsymbol{p}^n) ,~n \in \mathcal{N}$. Therefore, $Q^\text{min}_n=q_n (\boldsymbol{p^*}^n)=\sum\limits_{k \in \mathcal{K}_n} {p^n_{k}}^*$ is indeed the lower-bound of $q_n$. According to \eqref{Constraint mask q}, we guarantee that any $q_n \in \left[Q^\text{min}_n,P^{\text{mask}}_n\right],~\forall n \in \mathcal{N}$ satisfying \eqref{Constraint cell q} is feasible, and the proof is completed.
\end{proof}
The feasibility of problem \eqref{SCuser problem} can be directly determined by utilizing Lemma \ref{lemma feasiblecluster}. In the following, we find the closed-form of optimal intra-cluster power allocation as a linear function of $\boldsymbol{q}$.
\begin{proposition}\label{proposition jointcluster}
	For any given feasible $\boldsymbol{q}=[q_1,\dots,q_n]$, the optimal intra-cluster powers can be obtained by
	\begin{equation}\label{opt IC i}
		{p^n_{k}}^*=\left(\beta_{k} \prod\limits_{j \in \mathcal{K}_n \atop h^n_{j} < h^n_{k}} \left(1-\beta_{j}\right)\right) q_n + c^n_{k},~~\forall k \in \mathcal{K}_n\setminus\{\Phi_n\},
	\end{equation}
	and
	\begin{equation}\label{opt IC M}
		{p^n_{\Phi_n}}^*= \left(1 - \sum\limits_{i \in \mathcal{K}_n \atop h^n_{i} < h^n_{\Phi_n}} \beta_{i} \prod\limits_{j \in \mathcal{K}_n \atop h^n_{j} < h^n_{i}} \left(1-\beta_{j}\right)\right) q_n - \sum\limits_{i \in \mathcal{K}_n \atop h^n_{i} < h^n_{\Phi_n}} c^n_{i},
	\end{equation}
	where $\beta_{k} = \frac{2^{(R^{\text{min}}_{k}/W_s)}-1} {2^{(R^{\text{min}}_{k}/W_s)}}, ~\forall k \in \mathcal{K}$, 
	$c^n_{k}=\beta_{k} \left( \frac{1}{h^n_{k}} - \sum\limits_{j \in \mathcal{K}_n \atop h^n_{j} < h^n_{k}} \frac{\prod\limits_{l \in \mathcal{K}_n \atop h^n_{j} < h^n_{l} < h^n_{k}} \left(1-\beta_{l}\right) \beta_{j}} {h^n_{j}}\right),~ \forall n \in \mathcal{N},~k \in \mathcal{K}_n$.
\end{proposition}
\begin{proof}
	For any given feasible $\boldsymbol{q}$, \eqref{Constraint cell q} and \eqref{Constraint mask q} can be removed from \eqref{SCuser problem}.
	Then, the convex problem \eqref{SCuser problem} can be equivalently divided into $N$ intra-cluster convex power allocation subproblems, where for each subchannel $n$, we find the optimal powers according to Appendix B in \cite{rezvani2021optimal}. It is proved that at the optimal point of sum-rate maximization problem in fully SC-SIC, all the users with lower decoding order get power to only maintain their minimum rate demands.
\end{proof}
In contrast to fully SC-SIC in \cite{rezvani2021optimal}, in NOMA, there is a competition among cluster-head users to get the rest of the cellular power. According to \eqref{opt IC M}, the optimal power of the NOMA cluster-head user $\Phi_n$ can be obtained as a function of $q_n$ given by
\begin{equation}\label{opt IC M qmin}
	{p^n_{\Phi_n}}^*=\alpha_n q_n - c_n, \forall n \in \mathcal{N},
\end{equation}
where $\alpha_n=\left(1 - \sum\limits_{i \in \mathcal{K}_n \atop h^n_{i} < h^n_{\Phi_n}} \beta_{i} \prod\limits_{j \in \mathcal{K}_n \atop h^n_{j} < h^n_{i}} \left(1-\beta_{j}\right)\right)$ and $c_n=\sum\limits_{i \in \mathcal{K}_n \atop h^n_{i} < h^n_{\Phi_n}} c^n_{i}$ are nonnegative constants. According to Proposition \ref{proposition jointcluster} and \eqref{opt IC M qmin}, the optimal value of \eqref{obf SCuser problem} for the given $q_n$ can be formulated in closed form as
\begin{multline}\label{optimal value}
	R^n_\text{opt} (q_n)=\sum\limits_{k \in \mathcal{K}_n} R^n_{k} (\boldsymbol{p^*}^n) = \sum\limits_{k \in \mathcal{K}_n \atop k \neq \Phi_n} (R^{\text{min}}_{k}) + R^n_{\Phi_n} (q_n) 
	\\
	= \sum\limits_{k \in \mathcal{K}_n \atop k \neq \Phi_n} (R^{\text{min}}_{k}) + W_s \log_2\left( 1 + \left(\alpha_n q_n - c_n\right) h^n_{\Phi_n} \right).
\end{multline}
According to the above, \eqref{SCuser problem} is equivalently transformed to the following inter-cluster power allocation problem as
\begin{subequations}\label{SCuser problem eq} 
	\begin{align}\label{obf SCuser problem eq}
		\max_{\boldsymbol{q}}\hspace{.0 cm} 
		~~ & \sum\limits_{n \in \mathcal{N}} W_s \log_2\left( 1 + \left(\alpha_n q_n - c_n\right) h^n_{\Phi_n} \right)
		\\
		\text{s.t.}~~
		\label{Constraint cell eq}
		& \sum\limits_{n \in \mathcal{N}} q_n = P^{\text{max}},
		\\
		\label{Constraint mask eq}
		& q_n \in [Q^\text{min}_n,P^{\text{mask}}_n],~\forall n \in \mathcal{N}.
	\end{align}
\end{subequations}
In the objective function \eqref{obf SCuser problem eq}, we substituted the closed-form of optimal intra-cluster powers as a function of $\boldsymbol{q}$. Constraint \eqref{Constraint cell eq} is the cellular power constraint, and \eqref{Constraint mask eq} is to ensure that the individual minimum rate demand of all the users within each NOMA cluster is satisfied. The objective function \eqref{obf SCuser problem eq} is strictly concave, and the feasible set of \eqref{SCuser problem eq} is affine. Accordingly, problem \eqref{SCuser problem eq} is convex.
For convenience, let $\boldsymbol{\tilde{q}}=[\tilde{q}_n],\forall n \in \mathcal{N}$, where $\tilde{q}_n=q_n - \frac{c_n}{\alpha_n}$. Then, \eqref{SCuser problem eq} is equivalent to the following convex problem as
\begin{subequations}\label{eq1 problem} 
	\begin{align}\label{obf eq1 problem}
		\max_{\boldsymbol{\tilde{q}}}\hspace{.0 cm} 
		~~ & \sum\limits_{n \in \mathcal{N}} W_s \log_2\left( 1 + \tilde{q}_n H_n \right)
		\\
		\text{s.t.}~~
		\label{Constraint 1}
		& \sum\limits_{n \in \mathcal{N}} \tilde{q} = \tilde{P}^{\text{max}},
		\\
		\label{Constraint 2}
		& \tilde{q}_n \in [\tilde{Q}^\text{min}_n,\tilde{P}^{\text{mask}}_n],~\forall n \in \mathcal{N},
	\end{align}
\end{subequations}
where $H_n=\alpha_n h^n_{\Phi_n}$, $\tilde{P}^{\text{max}}=P^{\text{max}} - \frac{c_n}{\alpha_n}$, $\tilde{Q}^\text{min}_n=Q^\text{min}_n - \frac{c_n}{\alpha_n}$, and $\tilde{P}^{\text{mask}}_n=P^{\text{mask}}_n - \frac{c_n}{\alpha_n}$. The equivalent OMA problem \eqref{eq1 problem} can be solved by using the water-filling algorithm \cite{10.1155/2008/643081,8995606}. After some mathematical manipulations, the optimal $\tilde{q}^*_n$ can be obtained by
\begin{equation}\label{bisection opt form}
	\tilde{q}^*_n=
	\begin{cases}
		\frac{W_s/(\ln 2)}{\nu^*} - \frac{1}{H_n}, &\quad \left(\frac{W_s/(\ln 2)}{\nu^*} - \frac{1}{H_n}\right) \in [\tilde{Q}^\text{min}_n,\tilde{P}^{\text{mask}}_n]\\
		0, &\quad \text{otherwise}, \\ 
	\end{cases}
\end{equation}
such that $\boldsymbol{q}^*$ satisfies \eqref{Constraint 1}. Also, $\nu^*$ is the dual optimal corresponding to constraint \eqref{Constraint 1}.
The pseudo-code of the bisection method for finding $\nu^*$ is presented in Alg. \ref{Alg bisection}. 
\begin{algorithm}[tp]
	\caption{The bisection method for finding $\nu^*$ in \eqref{bisection opt form}.} \label{Alg bisection}
	\begin{algorithmic}[1]
		\STATE Initialize tolerance $\epsilon$, lower-bound $\nu_l$, upper-bound $\nu_h$, and maximum iteration $L$.
		\FOR {$l=1:L$}
		\STATE Set $\nu_m=\frac{\nu_l + \nu_h}{2}$.
		\\
		\STATE\textbf{if}~~$\sum\limits_{n \in \mathcal{N}} \max \left\{\tilde{Q}^\text{min}_n, \min \left\{ \left(\frac{W_s/(\ln 2)}{\nu_m} - \frac{1}{H_n}\right) , \tilde{P}^{\text{mask}}_n\right\} \right\} < \tilde{P}^{\text{max}}$~\textbf{then}~~~Set $\nu_h=:\nu_m$.
		\\
		\STATE\textbf{else}~~~Set $\nu_l=:\nu_m$.
		\\
		\STATE\textbf{end if}
		\STATE\textbf{if}~~$\frac{\tilde{P}^{\text{max}} - \sum\limits_{n \in \mathcal{N}} \max \left\{\tilde{Q}^\text{min}_n, \min \left\{ \left(\frac{W_s/(\ln 2)}{\nu_m} - \frac{1}{H_n}\right) , \tilde{P}^{\text{mask}}_n\right\} \right\}}{\tilde{P}^{\text{max}}} \leq \epsilon$~\textbf{then}~~~\textbf{break}.
		\\
		\STATE\textbf{end if}
		\ENDFOR
	\end{algorithmic}
\end{algorithm}
After finding $\boldsymbol{\tilde{q}}^*$, we obtain $\boldsymbol{q}^*$ by $q^*_n=(\tilde{q}^*_n + \frac{c_n} {\alpha_n}),~\forall n \in \mathcal{N}$. Then, we find the optimal powers according to Proposition \ref{proposition jointcluster}.

\begin{remark}\label{remark virtualuser}
	In the transformation of \eqref{pureNOMA problem} to \eqref{eq1 problem}, the NOMA system is equivalently transformed to a virtual OMA system including a single virtual BS with maximum power $\tilde{P}^{\text{max}}$, and $N$ virtual OMA users operating in $N$ subchannels with maximum allowable power $\tilde{P}^{\text{mask}}_n,~\forall n \in \mathcal{N}$. Each NOMA cluster $n \in \mathcal{N}$ is indeed a virtual OMA user whose channel gain is $H_n=\alpha_n h^n_{\Phi_n}$, where $\alpha_n$ is a decreasing function of the minimum rate demand of users with lower decoding order in NOMA cluster $n$, and the channel gain of the NOMA cluster-head user, whose index is $\Phi_n$. Each virtual user $n$ has also a minimum power demand $\tilde{Q}^\text{min}_n$ in order to guarantee the minimum rate demand of its multiplexed users in $\mathcal{K}_n$. For any given virtual power budget $\boldsymbol{\tilde{q}}$, the achievable rate of each virtual OMA user is sum-rate of its multiplexed users.
\end{remark}
The exemplary models of different multiple access techniques for a SISO Gaussian BC including KKT optimality conditions analysis as well as equivalent transformation to virtual OMA system are shown in Fig. \ref{Fig_KKTModel}.
\begin{figure*}[tp]
	\centering
	\subfigure[Equivalent model of fully SC-SIC: A single virtual user.]{
		\includegraphics[scale=0.3]{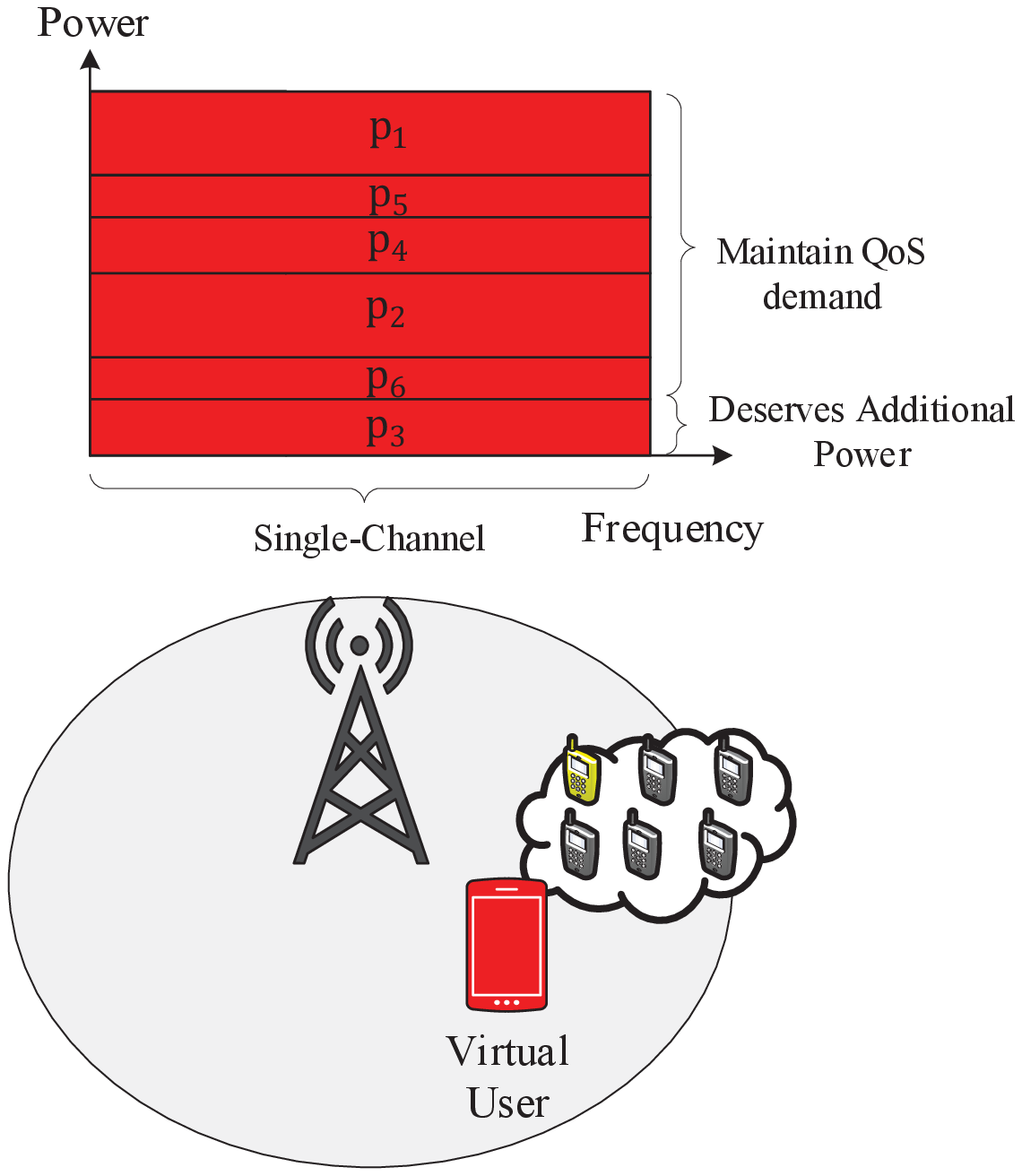}
		\label{Fully-SC-SIC-KKT}
	}\hfil
	\subfigure[Equivalent OMA model of NOMA: FDMA with $3$ virtual OMA users.]{
		\includegraphics[scale=0.3]{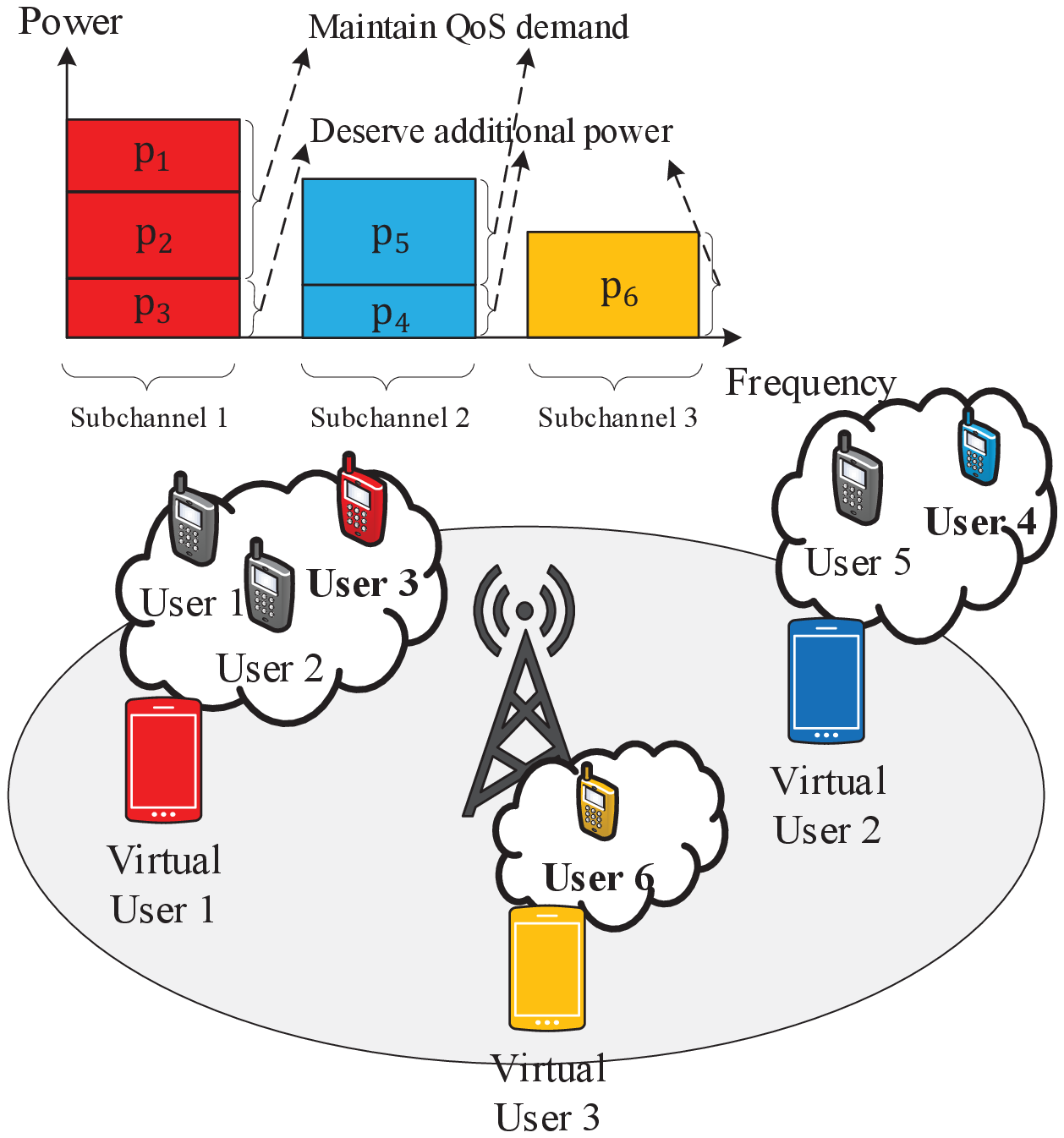}
		\label{NOMA-KKT}
	}\hfil
	\subfigure[FDMA: Real OMA users.]{
		\includegraphics[scale=0.3]{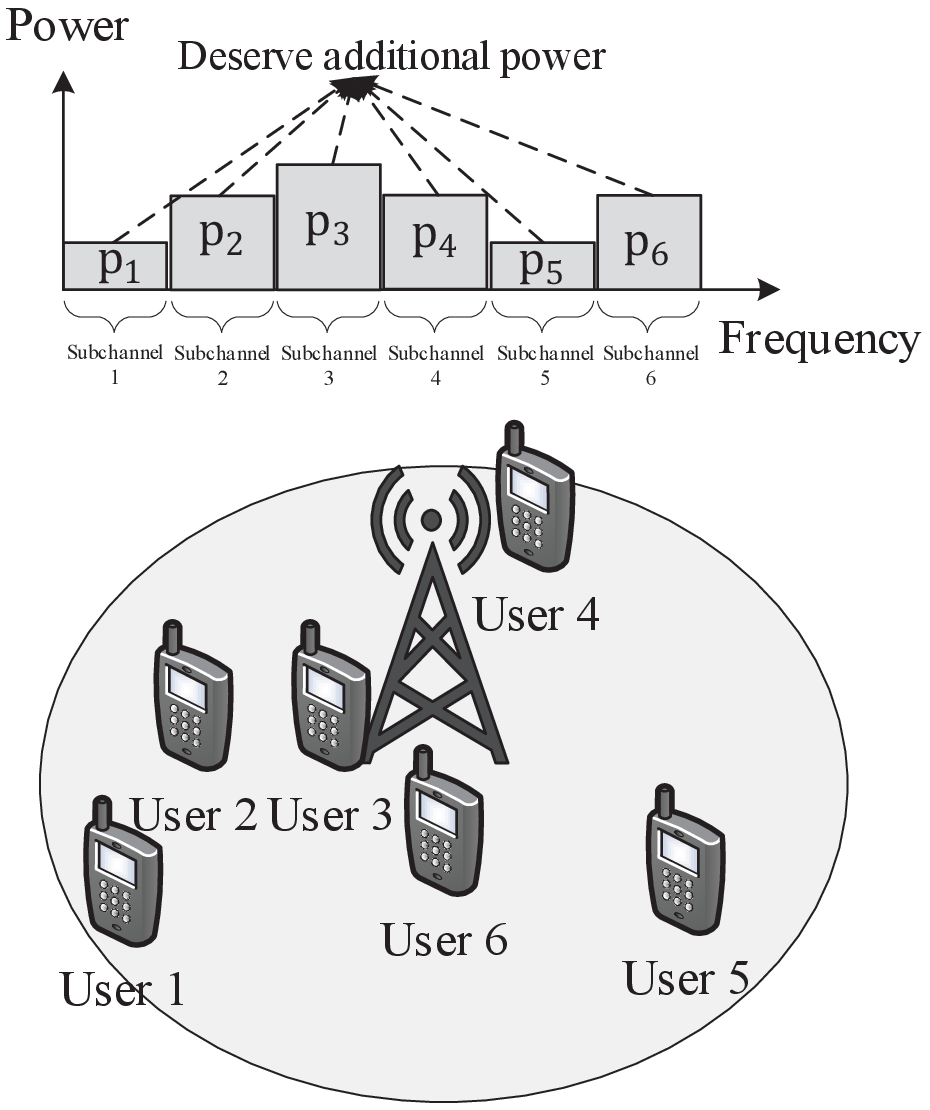}
		\label{OMA-KKT}
	}\hfil
	\caption
	{The equivalent models of fully SC-SIC and NOMA including virtual users. (a): User $3$ is cluster-head; (b): Users $3$, $4$ and $6$ are cluster-head; (c): All the users are cluster-head.}
	\label{Fig_KKTModel}
\end{figure*}
Note that FDMA is a special case of NOMA, where each subchannel is assigned to a single user. Hence, each OMA user acts as a NOMA cluster-head user, and subsequently, the virtual users are identical to the real OMA users, i.e., $\alpha_n=1$, $H_n=h^n_{\Phi_n}$, and $c_n=0$. As a result, each user in FDMA deserves additional power. It can be shown that the analysis for finding the optimal power allocation in NOMA and FDMA are quite similar, and the only differences are\footnote{Additional note: In \cite{rezvani2021optimal}, we showed that $c_n \approx 0$.} $\alpha_n$ and $c_n$. Finally, both of them can be solved by using the water-filling algorithm for any number of users and clusters. 

\section{Numerical Results}\label{Section simulation}
In this section, we apply the Monto-Carlo simulations over $20,000$ channel realizations per-scenario, to evaluate the performance of fully SC-SIC, NOMA with different $U^\text{max}$, and FDMA. The simulation settings are shown in Table \ref{Table parameters}.
\begin{table}[tp]
	\centering
	\caption{System Parameters}
	\begin{adjustbox}{width=\columnwidth,center}
		\begin{tabular}{c c}
			\hline
			\textbf{Parameter} & \textbf{Value} \\
			\hline
			BS maximum transmit power & 46 dBm \\
			Coverage of BS & Circular area with radii of $500$ m  \\
			Wireless Bandwidth & 5 MHz \\
			Number of users & $\{5,10,15,\dots,60\}$  \\
			User distribution model & Randomly (uniform distribution) \\
			$U^\text{max}$ in NOMA  & $\{2,4,6\}$ \\
			Minimum distance of users to BS & 20 m  \\
			Distance-depended path loss & $128.1 + 37.6 \log_{10} (d/\text{Km})$ dB  \\
			Lognormal shadowing standard deviation & $8$ dB  \\
			Small-scale fading & Rayleigh fading \\
			AWGN power density & -174 dBm/Hz \\
			Minimum rate demand of users  & $\{0.25,0.5,0.75,1,\dots,5\}$ Mbps \\
			\hline
		\end{tabular}
		\label{Table parameters}
	\end{adjustbox}
\end{table}
The optimal user grouping in NOMA is known to be NP-hard 
\cite{7263349,9154358,7982784,7557079}, thus exhasutive search is required to achieve the globally optimal user grouping. Due to the exponential complexity of exhaustive search in the number of users, similar to all the existing works, we apply a suboptimal user grouping, where each subchannel (with $N=\lceil K/U^\text{max}\rceil$) is assigned to the strongest candidate user. It is worth noting that the optimal power allocation provided in this work is valid for any given user grouping including arbitrary number of multiplexed users within each group, such as swap-matching based user grouping in \cite{7982784} (which is valid for only $2$-order NOMA clusters) or the heuristic strategies in \cite{7557079}. The term '$X$-NOMA' is referred to NOMA with $U^\text{max}=X$.

Fig. \ref{Figoutageminrate} evaluates the outage probability of fully SC-SIC, $X$-NOMA, and FDMA in various number of users and minimum rate demands. 
\begin{figure}
	\centering
	\subfigure[Outage probability vs. users $R^{\text{min}}_{k}$ for $K=30$.]{
		\includegraphics[scale=0.38]{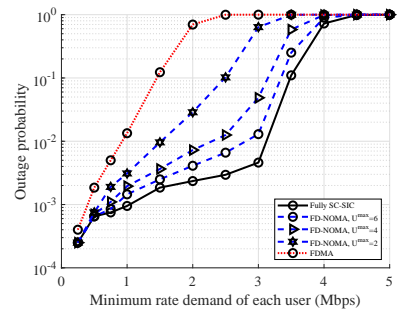}
		\label{Fig_outage_minrate2}
	}\hfil
	\subfigure[Outage probability vs. number of users for $R^{\text{min}}_{k}=3~\text{Mbps},\forall k \in \mathcal{K}$.]{
		\includegraphics[scale=0.38]{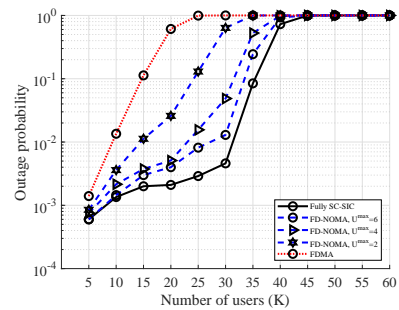}
		\label{Fig_outage_usernum2}
	}\hfil
	\caption
	{Impact of the minimum rate demand and number of users on outage probability of fully SC-SIC, NOMA, and FDMA.}
	\label{Figoutageminrate}
\end{figure}
The outage probability of each multiple access technique is calculated by dividing the number of infeasible solutions determined according to Lemma \ref{lemma feasiblecluster}, instances by total number of channel realizations. The results show that for quite lower levels of $R^{\text{min}}_{k}$, the performance gap between different multiple access techniques is low. For moderate $R^{\text{min}}_{k}$, there exist significant performance gaps between FDMA and $2$-NOMA, and also between $2$-NOMA and $4$-NOMA. However, the performance gap between $4$-NOMA and $6$-NOMA is quite low. The results show that multiplexing more than $2$ users results significant performance gain, such that this performance gain highly decreases when $U^\text{max}$ increases. Finally, for quite large $R^{\text{min}}_{k}$, the outage probability of all the techniques tends to $1$. Similar arguments also hold for different number of users shown in Fig. \ref{Fig_outage_usernum2}. In summary, the outage probability follows: $\text{Fully SC-SIC} < 6\text{-NOMA} \approx 4\text{-NOMA} < 2\text{-NOMA} \ll \text{FDMA}$.

Fig. \ref{FigSRminrate} evaluates the average sum-rate for different minimum rate demands as well as number of users.
\begin{figure}
	\centering
	\subfigure[Average sum-rate vs. $R^{\text{min}}_{k}$ for $K=30$.]{
		\includegraphics[scale=0.38]{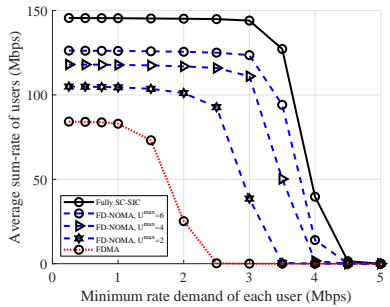}
		\label{Fig_SR_minrate2}
	}\hfil
	\subfigure[Average sum-rate vs. number of users for $R^{\text{min}}_{k}=3~\text{Mbps},\forall k \in \mathcal{K}$.]{
		\includegraphics[scale=0.38]{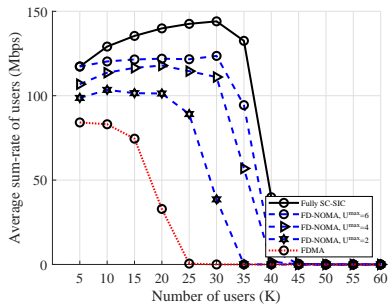}
		\label{Fig_SR_usernum2}
	}\hfil
	\caption
	{Impact of the minimum rate demand and number of users on average sum-rate of fully SC-SIC, NOMA, and FDMA.}
	\label{FigSRminrate}
\end{figure}
For the case that outage occurs, the sum-rate is set to zero.
As can be seen, there is a significant performance gap between FDMA and NOMA or fully SC-SIC for low and moderate $R^{\text{min}}_{k}$ and/or number of users. In this range, increasing the number of users increases the performance gains between FDMA, $X$-NOMA ($2 \leq X \leq 6$), and fully SC-SIC as is shown in Fig. \ref{Fig_SR_usernum2}. Moreover, fully SC-SIC well exploits the multiuser diversity specifically when $R^{\text{min}}_{k}$ and number of users are low. However, the sum-rate of $X$-NOMA remains approximately constant, specifically for $X=4,6$. Finally, increasing the number of users in $2$-NOMA and FDMA reduces sum-rate. Similar to outage probability in Fig. \ref{Figoutageminrate}, we observe that the performance gain in terms of sum-rate between $4$-NOMA and $6$-NOMA is quite low, and is decreasing between $X$-NOMA and $(X+1)$-NOMA, when $X \to K$. In summary, the sum-rate follows: $\text{Fully SC-SIC} > 6\text{-NOMA} \approx 4\text{-NOMA} > 2\text{-NOMA} \gg \text{FDMA}$.

\section{Concluding Remarks and Future Works}\label{Section conclusion}
In this paper, we proposed a fast water-filling algorithm to find the globally optimal power allocation of a general multiuser downlink NOMA system. In this algorithm, we utilized the formulated closed-form of optimal powers among multiplxed users within each group. We showed that network-NOMA with $N$ clusters can be equivalently transformed to a virtual network-OMA including $N$ virtual users, whose channel gains and minimum required powers are obtained in closed form. Numerical assessments show that there exist a considerable performance gap in terms of both the outage probability and sum-rate between FDMA and $2$-NOMA as well as between $2$-NOMA and $4$-NOMA. These performance gaps highly decrease when $U^\text{max}>4$.
A list of future works are provided as follows:
\begin{enumerate}
	\item Hybrid-NOMA: For the case that each user occupies more than one subcarrier, further efforts are needed to tackle the nonconvex per-user minimum rate constraint.
	\item Maximizing Energy-Efficiency: Our analysis can be applied to the fractional energy-efficiency maximization problem.
	\item Behind-NOMA: Although NOMA provides a feasible solution for the case that the number of multiplexed-users is limited, the user gouping of NOMA (based on OMA) is still suboptimal. In information theory, it is proved that fully SC-SIC is capacity achieving in SISO Gaussian BCs when $U^\text{max} \geq K$. However, \textit{"What is the capacity-achieving technique when $U^\text{max} < K$?"}
\end{enumerate}

\appendices

\bibliographystyle{IEEEtran}
\bibliography{IEEEabrv,Bibliography}

\end{document}